\PassOptionsToPackage{hidelinks}{hyperref}
\documentclass[draftclsnofoot]{IEEEtran}

\makeatletter
\let\NAT@parse\undefined
\makeatother

\let\labelindent\relax

\usepackage{cite}
\usepackage{graphicx}
\usepackage{mathtools,amsthm,amssymb}
\usepackage{nicematrix}
\usepackage{enumitem}
\usepackage{array}
\usepackage{xcolor}
\usepackage{hyphenat}
\usepackage[utf8]{inputenc} 
\usepackage[T1]{fontenc}
\usepackage{lmodern}
\theoremstyle{plain}
\newtheorem{theorem}{Theorem}
\newtheorem{lemma}[theorem]{Lemma}
\newtheorem{proposition}[theorem]{Proposition}
\newtheorem{corollary}[theorem]{Corollary}
\theoremstyle{definition}
\newtheorem{definition}[theorem]{Definition}
\newtheorem{example}{Example}
\theoremstyle{remark}
\newtheorem{remark}{Remark}
\newcommand{\lie}{\mathrm{Lie}\,}
\newcommand{\tr}{\intercal}
\newcommand{\ev}{\mathrm{ev}_m}
\DeclareMathOperator{\rank}{\mathrm{rank}}
\DeclarePairedDelimiter\norm{\lVert}{\rVert}
\usepackage{easyReview}
\usepackage{hyperref,bookmark}

\begin{document}
\title{Bilinear Systems Induced by Proper Lie Group Actions}
\author{Gong~Cheng,
        Wei~Zhang, 
        and Jr-Shin~Li 
\thanks{The authors are with the Department of Electrical and Systems Engineering, Washington University, Saint Louis, MO 63130, USA (e-mail: \{gong.cheng, wei.zhang, jsli\}{@wustl.edu}).}}

\maketitle

\begin{abstract}
	In the study of induced bilinear systems, the classical Lie algebra rank condition (LARC) is known to be impractical since it requires computing the rank everywhere. On the other hand, the transitive Lie algebra condition, while more commonly used, relies on the classification of transitive Lie algebras, which is elusive except for few simple geometric objects such as spheres. We prove in this note that for bilinear systems induced by proper Lie group actions, the underlying Lie algebra is closely related to the orbits of the group action. Knowing the pattern of the Lie algebra rank over the manifold, we show that the LARC can be relaxed so that it suffices to check the rank at an arbitrary single point. Moreover, it removes the necessity for classifying transitive Lie algebras. Finally, this relaxed rank condition also leads to a characterization of controllable submanifolds by orbits.
\end{abstract}


\section{Introduction}
\label{sec:intro}

Bilinear systems induced by Lie group actions have attracted great attention in control theory and engineering. Prominent examples range from excitation of spin systems in quantum physics \cite{Glaser98,Li_PRA06,Li_PNAS11} to control of rigid bodies in mechanics to manipulation of neuron oscillators in neuroscience \cite{Gerstner02,Dayan05}. In these examples, examining controllability of the considered systems is one of the most fundamental steps. However, the classical controllability condition, the Lie algebra rank condition (LARC), requires the examination on the rank of the underlying Lie algebra of the system at each point over the entire state space, thus impractical in most cases. On the other hand, it has been revealed that induced bilinear systems on homogeneous spaces are closely related to their counterparts on Lie groups, so that controllability of induced systems corresponds to transitivity of Lie group actions (or Lie subalgebras) \cite{Sachkov2009,Brockett73b}. Following this correspondence, many works have been done to study the controllability of induced systems using transitivity \cite{BOOTHBY/WILSON1979,BOOTHBY1975296}. One major obstacle in this approach is that it requires prior knowledge about the classification of transitive Lie subalgebras, which is only known for very few manifolds: Montgomery and Samelson \cite{Montgomery1943} classified all transitive Lie group actions on spheres in 1943, and Boothby \cite{BOOTHBY1975296,Elliott09} classified all transitive Lie algebras for punctured Euclidean spaces in 1975.

In this note, we investigate the Lie algebra rank and characterize its pattern from a group action viewpoint. By establishing a connection between the underlying Lie algebra of the induced system and the orbits of its corresponding group action, we show that the rank is constant on each orbit, given that the group action is proper. Consequently, we can relax the classical LARC such that it suffices to check the Lie algebra rank at an arbitrary single point, instead of over the entire state space. This relaxed condition applies to bilinear systems induced by any compact Lie group, as well as many non\hyp{}compact Lie groups such as the special Euclidean group $\mathrm{SE}(n)$. Our result simplifies the analysis of controllability of many induced bilinear systems, since classifications of transitive Lie algebras are no longer necessary. Furthermore, we also prove that for bilinear systems induced by proper Lie group actions that are not controllable, their controllable submanifolds can be fully characterized by the orbits of group actions.

The note is organized as follows: in Section~\ref{sec:pre} we provide an overview of the concepts of group actions, orbits, and induced bilinear systems. In Section~\ref{sec:son}, we formulate the relaxed rank condition for bilinear systems on $\mathbb{S}^{n-1}$ induced by $\mathrm{SO}(n)$. In Section~\ref{sec:proper.action}, we generalize our result to all bilinear systems induced by proper Lie group actions.

\section{Bilinear Systems Induced by Lie Group Actions}
\label{sec:pre}

The Lie algebra rank condition (LARC) marks an important milestone in geometric control theory aiming at increasing the accessibility to nonlinear systems by exploiting tools in modern mathematics, in particular, differential geometry and Lie theory. Technically, LARC reveals the connection between controllability of a nonlinear system and the rank of the Lie algebra generated by the vector fields which govern the system dynamics (See Appendix \ref{appd:LARC}). So far, LRAC remains one of the most widely-used criteria for examining controllability of nonlinear systems, but it also has its limitations. For example, as is shown in the motivating example below, systems with vector fields which generate a non-constant rank distribution are outside its scope.

Let us consider the bilinear system defined on $\mathbb{R}^n$ as
\begin{equation}
  \label{eq:induced.system.son}
  \dot{x}(t) = \Bigl(B_0+\sum_{i=1}^mu_i(t)B_i\Bigr)x(t), \quad
  x(0)=x_0\in\mathbb{R}^n\backslash\{0\},
\end{equation}
where $B_i\in\mathfrak{so}(n)$, the special orthogonal Lie algebra consisting of $n\times n$ anti-symmetric matrices, for all $i=0,\dots,m$. Note that every $B_i$ gives rise to an element $B_ix$ in $\mathfrak{X}(\mathbb{R}^n)$, the space of smooth vector fields on $\mathbb{R}^n$, and at any $x\in\mathbb{R}^n$ the tangent vector $B_ix\in \mathsf{T}_x\mathbb{R}^n$ satisfies the following identity
\[
  x\cdot{}(B_ix)=x^{\tr}B_ix=(x^{\tr}B_ix)^{\tr}=x^{\tr}(-B_i)x=0,
\]
where $\mathsf{T}_x\mathbb{R}^n$ denotes the tangent space of $\mathbb{R}^n$ at $x$, and ``$\cdot$'' is the Euclidean inner product. Geometrically, the above identity implies that the tangent vector $B_ix$ is always perpendicular to the vector $x$. Adopting this observation to the system in \eqref{eq:induced.system.son} yields
\[
  \frac{\mathrm{d}}{\mathrm{d}t}\|x(t)\|^2=2x(t)\cdot\dot x(t)=x(t)\cdot B_ix(t)=0.
\]
Namely, the Euclidean norm of the state variable is invariant under the system dynamics, so  the system is restricted to evolving on the sphere centered at the origin with radius $\|x_0\|$. Without loss of generality, we may assume for the initial condition $\|x_0\|=1$ and reduce the state space of the system to the $(n-1)$-dimensional unit sphere $\mathbb{S}^{n-1}$.


By the LARC, controllability of the system in \eqref{eq:induced.system.son} is determined by the rank of the Lie algebra $\mathfrak{L}=\lie\{B_{i}x:x\in\mathbb{S}^{n-1},\ i=0,\dots,m\}$. However, in general, $\mathfrak{L}$ is not of constant rank. For example, in the case of $n=4$, $m=3$, and $B_0=0$, 
\[
  B_1=
  \begin{pNiceMatrix}[columns-width=2mm]
    0 & & & \\
    & 0 & 1 & \\
    & -1 & 0 & \\
    & & & 0 \\
  \end{pNiceMatrix},\ %
  B_2=
  \begin{pNiceMatrix}[columns-width=2mm]
    0 & & & \\
    & 0 & & \\
    & & 0 & 1 \\ 
    & & -1 & 0 \\
  \end{pNiceMatrix},
\]
since $[B_1x,B_2x]=-[B_1,B_2]x$, we have that
\[
  \mathfrak{L}=\lie\{B_1x,B_2x\}=\mathrm{span}\,\{B_1x,B_2x,[B_1,B_2]x\},
\]
where
\[
  [B_1,B_2]=
  \begin{pNiceMatrix}[columns-width=2mm]
    0 & & & \\
    & 0 & & 1 \\
    & & 0 & \\
    & -1 & & 0 \\
  \end{pNiceMatrix}.
\]
We can check that at $x_1=(1,0,0,0)^\tr$, $B_1x_1=B_2x_1=[B_1x_1,B_2x_1]=0$, so $\rank\mathfrak{L}{}=0$ at $x_1$. In comparison, at $x_2=(0,1,0,0)^\tr$, $\mathfrak{L}$ is the span of $(0,0,-1,0)^{\tr}$ and $(0,0,0,-1)^{\tr}$, so $\rank\mathfrak{L}=2$ at $x_2$. As a result, the Lie algebra $\mathfrak{L}\subset\mathfrak{X}(\mathbb{S}^3)$ gives a non-constant rank distribution on $\mathbb{S}^3$. This poses a problem for determining the controllable submanifold of the system by using LARC. 

To elaborate this non-constant rank phenomenon from an algebraic viewpoint, we note that $\mathbb{S}^{n-1}$ is an ${\rm SO}(n)$-homogeneous space, so that the vector fields $\{B_ix:x\in\mathbb{S}^{n-1},\ i=0,\dots,m\}$ which govern the dynamics of the system in \eqref{eq:induced.system.son} are generated by the $\mathfrak{so}(n)$-action on $\mathbb{R}^n$. Equivalently, the system in \eqref{eq:induced.system.son} is induced by the Lie group action of SO$(n)$ on $\mathbb{R}^n$. However, the SO$(n)$-action on $\mathbb{S}^{n-1}$ is not free, and correspondingly, the $\mathfrak{so}(n)$-action as a Lie algebra homomorphism from $\mathfrak{so}(n)$ to $\mathfrak{X}(\mathbb{S}^{n-1})$ is not injective. So different elements in $\mathfrak{so}(n)$ may give the same vector fields on $\mathbb{S}^{n-1}$, which  result in the drop of the rank of the Lie algebra generated by $B_ix$. 


In what follows, to tackle the issue of non-constant rank, we integrate the theory on Lie group actions with geometric control theory to establish sufficient and necessary controllability conditions for bilinear systems induced by Lie group actions on smooth manifolds. The basics of Lie groups actions, homogeneous spaces and quotient manifolds that are necessary to our development are reviewed in Appendix~\ref{appd:group_action}.

\subsection*{Controllability of Induced Systems}

Next, we introduce the theory that connects induced bilinear systems to bilinear systems on Lie groups. First, we give the formal definition of induced systems.


\begin{definition}[Induced Systems]
  \label{def:induced.system}
  Let $G$ be a connected Lie group, $\mathfrak{g}$ its Lie algebra, and $M$ a smooth manifold which admits a left $G$-action. Given a right\hyp{}invariant bilinear system on $G$,
  \begin{equation}
    \label{eq:bilinear.system}
    \dot{X}=\Bigl(B_0+\sum_{i=1}^mu_i(t)B_i\Bigr)X(t), \quad
    B_i\in\mathfrak{g}, \quad X(0)=I,
  \end{equation}
  it induces a bilinear system on $M$ in the form of 
  \begin{equation}
    \label{eq:induced.system}
    \dot{x}(t)=(\theta_{\ast}B_0)(x)+\sum_{k=1}^mu_i(t)(\theta_{\ast}B_i)(x),
    \quad x(0)\in{}M.
  \end{equation}
  where $\theta_{\ast}B_i$ is the fundamental vector field of $B_i$ (see Definition~\ref{def:fundamental.vector.field}). The bilinear system in \eqref{eq:induced.system} is called the \emph{induced system} of \eqref{eq:bilinear.system}.
\end{definition}

The system in \eqref{eq:induced.system} is named as ``induced'' because its drifting and control vector fields are generated from those of the system in \eqref{eq:bilinear.system}, which reflect the action of $G$ on $M$ from the infinitesimal viewpoint. It is then natural that, other than the controllability of the system in \eqref{eq:bilinear.system} on $G$, the properties of the action of $G$ on $M$ also impacts on controllability of the induced system in \eqref{eq:induced.system}.


\begin{lemma}[Controllability of Induced Systems]
  \label{lem:controllability.induced}
  Let $M$ be a homogeneous $G$-space, i.e., $G$ acts transitively on $M$. Suppose that the bilinear system in \eqref{eq:bilinear.system} is controllable on $G$, or its reachable set $\mathcal{A}$, which is a semigroup, acts transitively on $M$, then the induced system in \eqref{eq:induced.system} is controllable on $M$.
\end{lemma}

\begin{proof}
  See \cite[Theorem~5.1]{Sachkov2009}.
\end{proof}

On the other hand, since all Lie groups and their homogeneous spaces are analytic, the following lemma provides a necessary condition for controllability of induced systems. In general, for a system
\begin{equation}
  \label{eq:system.rank.theorem}
  \dot{x}(t)=a_0(x)+\sum_{i=1}^{m}u_i(t)a_i(x)
\end{equation}
on an analytic manifold $M$ with $a_i(x)\in \mathfrak{X}(M)$ being analytic on $M$, we have the following result.

\begin{lemma}[See {\cite[Theorem~1]{ELLIOTT1971364}}]
  \label{lem:a.consequence.of.controllability}
  Suppose $M$ is an analytic manifolds of dimension $k$, and that its fundamental group $\pi_1(M)$ has no elements of infinite order. If a system in \eqref{eq:system.rank.theorem} is controllable on $M$, then the Lie subalgebra $\mathfrak{A}$, generated by $\{Y_{0}, a_1, \dots, a_{m}\}$ where $Y_{0}$ has the local form as
  \[
    Y_{0}=\frac{\partial}{\partial{}t}+\sum_{j}a^{j}_{0}\frac{\partial}{\partial{}x_{j}},
  \]
  has rank $k+1$ on the space $M\times\mathbb{R}$.
\end{lemma}

Since the vector fields $a_i(x)$ are time\hyp{}invariant and commute with the differential operator $\partial/\partial t$, if $\rank \mathfrak{A} = k-1$, then the Lie algebra generated by $\{a_0, \ldots, a_m\}$ must have rank $k$. So if we let $\iota:M\hookrightarrow M\times\mathbb{R}$ to denote the the inclusion map and taking $a_i = \iota_{\ast}\theta_{\ast}B_i$, we conclude the lemma below.

\begin{lemma}
  \label{lem:induced.rank.necessary}
  If the induced system in \eqref{eq:induced.system} is controllable on $M$, then the Lie algebra generated by the fundamental vector fields $\{\theta_{\ast}B_i,\ {}i=0,\dots,m\}$ is of constant rank $k$ on $M$.
\end{lemma}

\section{Bilinear Systems Induced by \texorpdfstring{{$\mathrm{SO}(n)$}}{SO(n)}-Actions and the Rank Condition} 
\label{sec:son}

In this section, we examine the bilinear system in \eqref{eq:induced.system.son} using the results of group actions that we introduced in Section~\ref{sec:pre}, and prove a sufficient and necessary controllability condition. First, we should emphasize that the system in \eqref{eq:induced.system.son} is an induced system. Indeed, for a natural group action of any matrix group $G$ on $\mathbb{R}^n$ (i.e., matrix multiplication), its fundamental vector fields on $\mathbb{R}^n$ are in the form of products: for any $B\in\mathfrak{g}$, let $(\theta_{\ast}B)$ denote the fundamental vector field on $\mathbb{R}^n$ corresponding to $B$, then by Definition~\ref{def:fundamental.vector.field} we have
\begin{equation}
  \label{eq:fundamental.vector.field(matrix)}
  \begin{aligned}
    (\theta_{\ast}B)(x) &=\frac{\mathrm{d}}{\mathrm{d}t}\Big|_{t=0}\exp(tB)x \\
    &= \Bigl(\frac{\mathrm{d}}{\mathrm{d}t}\Big|_{t=0}\exp(tB)\Bigr)x=Bx.
  \end{aligned}
\end{equation}
So according to Definition~\ref{def:induced.system}, the system in \eqref{eq:induced.system.son} is induced by the following bilinear system on $\mathrm{SO}(n)$:
\begin{equation}
  \label{eq:bilinear.system.son}
  \dot{X}=\Bigl(B_0+\sum_{i=1}^mu_i(t)B_i\Bigr)X(t), \quad X(0)=I.
\end{equation}

Recall from Lemma~\ref{lem:controllability.induced} that controllability of an induced systems is related to controllability of its counterpart on a Lie group. Therefore, the system in \eqref{eq:induced.system.son} on $\mathbb{S}^{n-1}$ is controllable if the bilinear system in \eqref{eq:bilinear.system.son} is controllable on $\mathrm{SO}(n)$, or if its controllable submanifold $H$, which is a closed Lie subgroup of $\mathrm{SO}(n)$ with its Lie subalgebra being $\mathfrak{h}:=\lie\{B_i\}\leqslant{}\mathfrak{so}(n)$, acts transitively on $\mathbb{S}^{n-1}$. So the controllability of \eqref{eq:induced.system.son} depends on the transitivity of the action by $H$, and for that, we introduce a rank condition that connects the transitivity of $H$-action to the rank of $\mathfrak{L}:=\lie\{B_ix\}$, the underlying Lie algebra of \eqref{eq:induced.system.son}. We will see that, in this case, it suffices to check the rank of $\mathfrak{L}$ at \emph{any single} point $x$ on the sphere $\mathbb{S}^{n-1}$.

\begin{definition}
  \label{def:rank.of.lie.algebra}
  Let $G$ be a matrix group acting naturally on $\mathbb{R}^n$, and $\mathfrak{g}$ its Lie algebra. The \emph{rank} of $\mathfrak{g}$ at a point $x\in\mathbb{R}^n$ is defined as the dimension of the subspace $\{Ax:A\in\mathfrak{g}\}\subseteq\mathsf{T}_{x}\mathbb{R}^n$, i.e.,
  \[
    \rank_{x}\mathfrak{g}:=\dim\{Ax:A\in\mathfrak{g}\}.
  \]
\end{definition}

Note that by Definition~\ref{def:rank.of.lie.algebra}, the rank of the Lie algebra $\mathfrak{L}$ generated by vector fields $\lie\{B_ix\}$ coincides with $\rank_{x}\mathfrak{h}$ at all $x$, where $\mathfrak{h}=\lie\{B_i\}$. The next lemma shows a group action of a subgroup of $\mathrm{SO}(n)$ is transitive if its Lie algebra has the \emph{maximal} rank.

\begin{lemma}
  \label{lem:transitive.lie.algebra}
  Let $\mathfrak{h}$ be a Lie subalgebra of $\mathfrak{so}(n)$ which corresponds to a closed Lie subgroup $H\leqslant{}\mathrm{SO}(n)$. The (natural) group action of $H$ on $\mathbb{S}^{n-1}$ is \emph{transitive} if and only if $\rank_{x}\mathfrak{h}=n-1$ for some $x\in\mathbb{S}^{n-1}$.
\end{lemma}

\begin{proof} 
  As we will show in Lemma~\ref{lem:constant.rank.on.orbits}, the rank of $\mathfrak{h}$ at $x$ coincides with the dimension of the orbit $H(x)\subset\mathbb{S}^{n-1}$. Suppose $H$ acts transitively on $\mathbb{S}^{n-1}$, i.e., $H(x)=\mathbb{S}^{n-1}$ for any $x\in\mathbb{S}^{n-1}$, so $\rank_{x}\mathfrak{h}=\dim{}H(x)=n-1$. On the other hand, if for some $x\in\mathbb{S}^{n-1}$ we have $\rank_{x}\mathfrak{h}=n-1$, then the orbit $H(x)\subseteq\mathbb{S}^{n-1}$, as a submanifold of dimension $n-1$, is open in $\mathbb{S}^{n-1}$. To show that the $H$-action is transitive (equivalently, $H(x)=\mathbb{S}^{n-1}$), it remains to show that $H(x)$ is also closed. This is because the $H$-action is proper for any closed $H\leqslant\mathrm{SO}(n)$, so its orbit $H(x)$ is closed in $\mathbb{S}^{n-1}$, by Lemma~\ref{lem:orbits.of.a.proper.action}. Therefore, since $\mathbb{S}^{n-1}$ is connected, we conclude that $H(x)=\mathbb{S}^{n-1}$ if $\rank_{x}\mathfrak{h}=n-1$.
\end{proof}

Following Lemma~\ref{lem:controllability.induced} and Lemma~\ref{lem:transitive.lie.algebra}, the system in \eqref{eq:induced.system.son} is controllable if $\rank_{x}\mathfrak{h}=n-1$ for some $x\in\mathbb{S}^{n-1}$, where $\mathfrak{h}=\lie\{B_i\}$. Conversely, we learn from Lemma~\ref{lem:induced.rank.necessary} that the rank condition is also sufficient. Combined with Lemma~\ref{lem:controllability.induced} and \ref{lem:transitive.lie.algebra}, we conclude that, since $\pi_1(\mathbb{S}^{n-1})=\{0\}$ for $n\geqslant{}3$, the system in \eqref{eq:induced.system.son} is controllable if and only if $\rank_{x}\mathfrak{h}=n-1$ for some $x\in\mathbb{S}^{n-1}$, where $\mathfrak{h}=\lie\{B_i\}$. When $n=2$, since the unit circle $\mathbb{S}^{1}$ is one-dimensional, controllability naturally guarantees the rank condition. In summary, we have the following theorem for controllability of the system in \eqref{eq:induced.system.son}.

\begin{theorem}
  \label{thm:equivalent.controllability.induced}
  Let $\mathfrak{h}:=\lie\{B_0,\ldots,B_m\}$ be a Lie subalgebra of $\mathfrak{so}(n)$ and $H\leqslant{}\mathrm{SO}(n)$ be a closed Lie subgroup whose Lie algebra equals to $\mathfrak{h}$. The following statements are equivalent:
  \begin{enumerate}[font=\normalfont,label={(\arabic*)}]%
    \item the system in \eqref{eq:induced.system.son} is controllable on $\mathbb{S}^{n-1}$;
    \item $\rank_{x}\mathfrak{h}=n-1$ for some $x\in\mathbb{S}^{n-1}$;
    \item $H$ acts transitively on $\mathbb{S}^{n-1}$.
  \end{enumerate}
\end{theorem}

\begin{example}

  Let us consider the Bloch system, which is a bilinear system of the form,
  \begin{equation}
    \label{eq:Bloch}
    \dot x(t)=\bigl(\omega_0\Omega_z+u(t)\Omega_y+v(t)\Omega_z\bigr) x(t), \quad x(0)=x_0,
  \end{equation}
  where $x(t)$ denotes the bulk magnetic moment of the spin ensemble on the 2\hyp{}dimensional sphere $\mathbb{S}^2$, $\omega_0$ is the Larmor frequency, $u(t)$ and $v(t)$ are the external radio\hyp{}frequency fields (controls) applied on the $y$ and $z$ directions, respectively, and
  \[
    \Omega_x=\begin{pNiceMatrix}
      0 & 0 & 0 \\ 0 & 0 & -1 \\ 0 & 1 & 0
    \end{pNiceMatrix}, \ %
    \Omega_y=\begin{pNiceMatrix}
      0 & 0 & 1 \\ 0 & 0 & 0 \\ -1 & 0 & 0
    \end{pNiceMatrix}, \ %
    \Omega_z=\begin{pNiceMatrix}
      0 & -1 & 0 \\ 1 & 0 & 0 \\ 0 & 0 & 0
    \end{pNiceMatrix}
  \]
  in $\mathfrak{so}(3)$ are the generators of rotation around the $x$-, $y$-, and $z$-axis, respectively. From our earlier discussion, \eqref{eq:Bloch} is induced by the Bloch system defined on $\mathrm{SO}(3)$,
  \[
    \dot{X}(t)=\bigl(\omega_0\Omega_z+u(t)\Omega_y+v(t)\Omega_z\bigr) X(t) ,\quad X(0)=I.
  \]

  More generally, we can consider the following bilinear system
  \begin{equation}
    \label{eq:induced.system.standard.basis}
    \dot{x}(t)=\Bigl(\Omega_{i_0j_0}+\sum_{k=1}^{m}u_k(t)\Omega_{i_kj_k}\Bigr)x(t),
    \quad{} x(0)\in\mathbb{S}^{n-1},
  \end{equation}
  where $\Omega_{ij}=E_{ij}-E_{ji}$ is an anti-symmetric matrix belonging to the standard basis of $\mathrm{SO}(n)$. It is induced by
  \begin{equation}
    \label{eq:bilinear.system.standard.basis}
    \dot{X}(t)=\Bigl(\Omega_{i_0j_0}+\sum_{k=1}^{m}u_k(t)\Omega_{i_kj_k}\Bigr)X(t),
    \quad{} X(0)\in\mathrm{SO}(n),
  \end{equation}
  on $\mathrm{SO}(n)$. It is proved in \cite{Cheng2021,Zhang19} that the bilinear systems in \eqref{eq:bilinear.system.standard.basis} is controllable on $\mathrm{SO}(n)$ if and only if the associated graph $\mathcal{G}=(\mathcal{V},\mathcal{E})$ is connected, where $\mathcal{V}=\{v_1,\ldots,v_n\}$ and $\mathcal{E}=\{v_{i_0}v_{j_0}, \ldots, v_{i_m}v_{j_m}\}$. Next, we will show that the controllability of \eqref{eq:induced.system.standard.basis} and \eqref{eq:bilinear.system.standard.basis} coincide, so graph connectivity criterion also applies to \eqref{eq:induced.system.standard.basis}.

  \begin{proposition}
    \label{prop:induced.controllability.connectivity}
    An induced system in \eqref{eq:induced.system.standard.basis} is controllable on $\mathbb{S}^{n-1}$ if and only if $\mathcal{G}$ is a connected graph.
  \end{proposition}
  
  \begin{proof}
    The part of sufficiency is clear, since the connectivity of $\mathcal{G}$ ensures controllability of the bilinear system on $\mathrm{SO}(n)$, so by Lemma~\ref{lem:controllability.induced}, the induced system in \eqref{eq:induced.system.standard.basis} is also controllable.  On the other hand, without loss of generality, we may assume that the vertex $v_n$ is \emph{not} contained in $\mathcal{G}$, which means $\Omega_{jn}\not\in\mathfrak{g}$ for any $j=1,\ldots,n-1$, where $\mathfrak{g}\leqslant{}\mathfrak{so}(n)$ is the Lie subalgebra generated by $\{\Omega_{i_0j_0}, \Omega_{i_1j_1}, \ldots, \Omega_{i_mj_m}\}$. Therefore, each group element in $G=\langle\exp\mathfrak{g}\rangle$ has the form $\begin{pNiceMatrix}[small] A & \\ & 1 \end{pNiceMatrix}$ with $A$ being a matrix in $\mathrm{SO}(n-1)$, which implies that the $G$-action fixes the point $x_n=(0,\ldots,1)^{\tr}\in\mathbb{S}^{n-1}$, and thus is \emph{not} transitive. So by the necessity condition in Theorem~\ref{thm:equivalent.controllability.induced}, system in \eqref{eq:induced.system.standard.basis} is not controllable if the graph $\mathcal{G}$ is not connected.
  \end{proof}
\end{example}

So far we have proved a rank condition which is both necessary and sufficient for the controllability of systems induced by $\mathrm{SO}(n)$-actions, so that the Lie algebra $\mathfrak{L}$ of drifting and control vector fields of a \emph{controllable} system in \eqref{eq:induced.system.son} is of \emph{constant full rank}. Next, we turn our focus to induced systems that are not controllable, and describe their controllable submanifolds. 

\begin{example}
  \label{ex:non-constant.rank.vector.field}
  Let us revisit the bilinear system we introduced at the beginning of Section~\ref{sec:pre} evolving on $\mathbb{S}^3$ which is induced by the $\mathrm{SO}(4)$-action with $m=2$, $\norm{x_0}=1$, where
  \[
    B_0=0,\ %
    B_1=
    \begin{pNiceMatrix}
      0 & & & \\
      & 0 & 1 & \\
      & -1 & 0 & \\
      & & & 0 \\
    \end{pNiceMatrix},\ %
    B_2=
    \begin{pNiceMatrix}
      0 & & & \\
      & 0 & & \\
      & & 0 & 1 \\
      & & -1 & 0 \\
    \end{pNiceMatrix}.
  \]
  Using Proposition~\ref{prop:induced.controllability.connectivity}, we find that the system is \emph{not} controllable, because the associated graph $\mathcal{G}$ is not connected.
  
  We have demonstrated that the Lie algebra $\mathfrak{L}\subset\mathfrak{X}(\mathbb{S}^3)$ generated by vector fields $\{B_1x, B_2x\}$ is not of constant rank: $\rank_{x_1}\mathfrak{g}=0$ for $x_1=(1,0,0,0)^{\tr}$, and $\rank_{x_2}\mathfrak{g}=\rank_{x_3}=2$ for $x_2=(0,1,0,0)^{\tr}$ and $x_3=(0,0,1,0)^{\tr}$. To see how the rank of the algebra $\mathfrak{L}$ changes, from the perspective of group action, we notice that $x_1=(1,0,0,0)^{\tr}$ and $x_2=(0,1,0,0)^{\tr}$ belong to different orbits. To be more specific, the controllable submanifold of the corresponding bilinear system
  \[
    \dot{X}(t)=\bigl(u_1(t)B_1+u_2(t)B_2\bigr)X, \quad X(0)=I
  \]
  on $\mathrm{SO}(4)$ is the closed Lie subgroup $H<\mathrm{SO}(4)$, which is generated by the Lie subalgebra $\mathfrak{h}:=\lie\{B_0,B_1,B_2\}$. $H$ consists of all matrices of the form \(\begin{psmallmatrix} 1 & \\ & A \end{psmallmatrix}\), where $A\in\mathrm{SO}(3)$ is a $3$-by-$3$ matrix. Therefore, $x_2=(0,1,0,0)^\tr$ and $x_3=(0,0,1,0)^\tr$ belong to the same orbit of $H$, so that the rank of $\mathfrak{L}$ on both points coincide (see Lemma~\ref{lem:constant.rank.on.orbits} below); while $x_1=(1,0,0,0)^{\tr}$ is in a different orbit. 
\end{example}

The above observation motivates us to examine how rank of a Lie subalgebra varies on a manifold, which lays the foundation for the study of controllable submanifolds. In fact, the following lemma asserts that the rank of $\mathfrak{L}$ is \emph{constant} over an orbit of $H$.

\begin{lemma}
  \label{lem:constant.rank.on.orbits}
  For a subalgbra $\mathfrak{h}:=\lie\{B_i\}\leqslant\mathfrak{so}(n)$ and let $H\leqslant{}\mathrm{SO}(n)$ be a closed Lie subgroup corresponding to $\mathfrak{h}$. Given $x\in\mathbb{S}^{n-1}$, the rank of $\mathfrak{h}$ is constant over the orbit $H(x)$, i.e., $\rank_{y}\mathfrak{h}=\dim{}H(x)$ for all $y\in{}H(x)$. Consequently, the rank of $\mathfrak{L}=\lie\{B_ix\}$ is \emph{constant} on $H (x)$.
\end{lemma}

\begin{proof}
  We know from Lemma~\ref{lem:orbits.of.a.proper.action} that the orbit $H(x)$ is a closed, embedded submanifold of $\mathbb{S}^{n-1}$, and that for any $y\in{}H(x)$,
  \begin{equation}
    \mathsf{T}_{y}H(x)=\{Ay:A\in\mathfrak{h}\}.
  \end{equation}
  Hence we have $\rank_{y}\mathfrak{h}=\dim{}H(x)$ for all $y\in{}H(x)$. As a consequence, since the rank of $\mathfrak{L}$ at $y\in{}H(x)$ equals to $\rank_{y}\mathfrak{h}$, $\mathfrak{L}$ is of constant rank over any given orbit $H(x)$.
\end{proof}

From Lemma~\ref{lem:constant.rank.on.orbits}, we can draw a clear link between controllable submanifolds and group orbits. More specifically, in the next theorem we prove that the controllable submanifold of system~\eqref{eq:induced.system.son} with a given initial condition $x(0)=x_0$ is a group orbit of $x_0$.

\begin{theorem}
  \label{thm:induced.system.controllable.submanifold}
  Let $\mathfrak{h}$ and $H$ be the same as in Lemma~\ref{lem:constant.rank.on.orbits}. For an induced bilinear system in \eqref{eq:induced.system.son} with initial condition $x(0)=x_0\in\mathbb{S}^{n-1}$, its controllable submanifold is the orbit $H(x_0)$.
\end{theorem}

\begin{proof}
  By \eqref{eq:fundamental.vector.field(matrix)} and Lemma~\ref{lem:orbits.of.a.proper.action}, the system in \eqref{eq:induced.system.son} evolves on the orbit $H(x_0)$, so its attainable set is a subset of $H(x_0)$. On the other hand, since $H$ acts transitively on the orbit $H(x_0)$ by definition, Lemma~\ref{lem:controllability.induced} guarantees that the system in \eqref{eq:induced.system.son} is controllable on $H(x_0)$. Therefore, we conclude that the controllable submanifold of the system in \eqref{eq:fundamental.vector.field(matrix)} is the orbit $H(x_0)$.
\end{proof}

\begin{remark}
  As a consequence of Theorem~\ref{thm:induced.system.controllable.submanifold}, we remark that the controllable submanifolds of induced systems in the form of~\eqref{eq:induced.system.son} can be classified by closed subgroups of $\mathfrak{so}(n)$ and their orbits. To be more specific, suppose $G\leqslant\mathrm{SO}(n)$ is a closed subgroup, $\mathfrak{g}=\mathsf{T}_{e}G$ is the Lie subalgebra of $\mathfrak{so}(n)$ corresponding to $G$ and let $\{A_1,\ldots,A_m\}$ be a (finite) set of generators of $\mathfrak{g}$, then the orbit $G(y_0)$ is the controllable submanifold of the following system
  \[
    \dot{x}(t)=\Bigl(\sum_{i=1}^{m} u_i(t)A_i\Bigr)x(t),\quad{}x(0)=y_0.
  \]
\end{remark}

In summary, in this section, we examine the bilinear systems induced by $\mathrm{SO}(n)$ group actions, and show that the rank of its underlying Lie algebra $\mathfrak{L}$ is constant over its controllable submanifold. Therefore, such a bilinear system is controllable if the rank of $\mathfrak{L}$ is maximal at any point; and if it is not controllable, the controllable submanifold is an orbit of the group action.

\section{Systems Induced by Proper Lie Group
  Actions} \label{sec:proper.action}

Following the discussions in previous sections on systems induced by the natural $\mathrm{SO}(n)$-action, in this section, we generalize Theorem~\ref{thm:equivalent.controllability.induced} and Theorem~\ref{thm:induced.system.controllable.submanifold}, and prove that the previous results still hold for systems induced by \emph{proper} Lie group actions.

Let $G$ be a (connected) matrix group acting \emph{properly} on $\mathbb{R}^n$, $\mathfrak{g}$ the Lie algebra of $G$, and $\mathcal{H}\subseteq\mathbb{R}^n$ a homogeneous $G$-space. For an induced bilinear system in the form of
\begin{equation}
  \label{eq:induced.system.driftless}
  \dot{x}(t)=\Bigl(\sum_{i=1}^m u_i(t)B_i\Bigr)x(t), \quad{}x(0)=x_0\in\mathcal{H},
\end{equation}
where $B_1,\ldots,B_m\in\mathfrak{g}$, we have the following theorem concerning the controllability of system~\eqref{eq:induced.system.driftless} and its controllable submanifolds.

\begin{theorem}
  \label{thm:induced.proper.action.rank.condition}
  Let $\mathfrak{h}=\lie\{B_1,\ldots,B_m\}$ be a Lie subalgebra of $\mathfrak{g}$ generated by matrices $B_i$ in \eqref{eq:induced.system.driftless}, and let $H\leqslant{}G$ be a closed Lie subgroup of $G$ corresponding to $\mathfrak{h}$. Suppose either (i) $H$ acts transitively on $\mathcal{H}$, or (ii) $\rank_{x}\mathfrak{h}=\dim\mathcal{H}$ for some $x\in\mathcal{H}$, then (iii) the system \eqref{eq:induced.system.driftless} is controllable. Additionally, if the fundamental group $\pi_1(\mathcal{H})$ has no elements of infinite order, then conditions (i), (ii) and (iii) are equivalent.  Furthermore, in case that system~\eqref{eq:induced.system.driftless} is not controllable, its controllable submanifold is the orbit $H(x_0)$.
\end{theorem}

\begin{proof}
  $\mathrm{(i)}\Rightarrow\mathrm{(iii)}$ comes directly from Lemma~\ref{lem:controllability.induced}. Note that $\mathcal{H}$ is connected since $H$ is connected, one can show $\mathrm{(i)}\Leftrightarrow\mathrm{(ii)}$ by an almost identical argument in the proof of Lemma~\ref{lem:transitive.lie.algebra}. In the case that $\pi_{1}(\mathcal{H})$ has no element of infinite order, one can apply Lemma~\ref{lem:a.consequence.of.controllability} to show that $\mathrm{(iii)}\Rightarrow\mathrm{(ii)}$, since the homogeneous space $\mathcal{H}$ is naturally analytic.

  Moreover, in case that system~\eqref{eq:induced.system.driftless} is non-controllable, Lemma~\ref{lem:orbits.of.a.proper.action} shows that the rank of $\mathfrak{h}$ is constant over the orbit $H(x_0)$, and that vector fields $\{B_ix\}$ are tangent to $H(x_0)$. Therefore, transitivity gives that $H(x_0)$ is the embedded controllable submanifold in $\mathcal{H}$ of system~\eqref{eq:induced.system.driftless}.
\end{proof}

\begin{remark}
  It is clear that Theorem~\ref{thm:induced.proper.action.rank.condition} applies to bilinear systems with drifting terms in the form of
  \[
    \dot{x}(t)=\Bigl(B_0+\sum_{i=1}^mu_i(t)B_i\Bigr)x(t),\quad{}x(0)=x_0\in\mathcal{H}
  \]
  if $G$ is compact, or if $\dot{X}(t)=B_0X(t)$, $X(0)=I$ has a periodic solution on $G$.
\end{remark}

Since all actions by compact groups are proper, Theorem~\ref{thm:induced.proper.action.rank.condition} applies to all systems induced by actions of compact Lie groups. Next, as an application of Theorem~\ref{thm:induced.proper.action.rank.condition} for non-compact groups, we show a rank condition for systems induced by $\mathrm{SE}(n)$-actions: consider a bilinear system in the form of
\begin{equation}
  \label{eq:driftless.induced.sen}
  \dot{x}(t)=\Bigl(\sum_{t=1}^{m}u_i(t)B_i\Bigr)x(t),\quad{}x(0)=x_0\in\mathbb{R}^n,
\end{equation}
where $B_i=(A_i,\mu_i)\in\mathfrak{se}(n)$ for $A_i\in\mathfrak{so}(n)$, $\mu_i\in\mathbb{R}^n$, and $B_ix:=A_ix+\mu_i$. Note that $\mathbb{R}^n$ is a homogeneous $\mathrm{SE}(n)$-space, and that the $\mathrm{SE}(n)$-action on $\mathbb{R}^n$ is proper.

\begin{lemma}
  \label{lem:proper.sen.action}
  The group action of $\mathrm{SE}(n)$ on $\mathbb{R}^n$ is proper.
\end{lemma}

\begin{proof}
  Consider the map $\alpha:\mathrm{SE}(n)\times{}\mathbb{R}^n\to\mathbb{R}^n\times{}\mathbb{R}^n$ s.t.\  $\alpha((R,\tau),x)=(x,Rx+\tau)$, where $R\in\mathrm{SO}(n)$, $\tau\in\mathbb{R}^n$, and $(R,\tau)\in\mathrm{SE}(n)$. To show that $\alpha$ is proper, without loss of generality, we need to show that the pre-image
  \begin{equation}
    \label{eq:pre-image.properness.proof}
    \{((R,\tau),x)\in\mathrm{SE}(n)\times\mathbb{R}^n: \norm{x}\leqslant{}c, \norm{Rx+\tau}\leqslant{}c\}
  \end{equation}
  is compact, for any constant $c$. Note that $\norm{Rx}=\norm{x}$ as $R\in\mathrm{SO}(n)$, $\norm{Rx+\tau}\leqslant{}c$ implies that $\norm{\tau}\leqslant{}2c$. Since
  \[
    \{(R,\tau)\in\mathrm{SE}(n):R\in\mathrm{SO}(n), \norm{\tau}\leqslant{}2c\}
  \]
  is a compact subset, the pre-image in \eqref{eq:pre-image.properness.proof} is also compact, which finishes our proof.
\end{proof}
Consequently, by Theorem~\ref{thm:induced.proper.action.rank.condition}, we have the following corollary for the system in \eqref{eq:driftless.induced.sen}:

\begin{corollary}
  \label{cor:induced.sen}
  Let $\mathfrak{g}=\lie\{B_1,\ldots,B_m\}$ be the Lie subalgebra of $\mathfrak{se}(n)$ generated by $B_i$ in \eqref{eq:driftless.induced.sen}, and let $G\leqslant{}\mathrm{SE}(n)$ denote a closed Lie subgroup corresponding to $\mathfrak{g}$. The following statements are equivalent:
  \begin{enumerate}[font=\normalfont,label={(\arabic*)}]%
    \item the system of the form \eqref{eq:driftless.induced.sen} is controllable on $\mathbb{R}^n$;
    \item $\rank_{x}\mathfrak{g}=n$ for some $x\in\mathbb{R}^n$;
    \item $G$ acts transitively on $\mathbb{R}^n$.
  \end{enumerate}
\end{corollary}

\section{Conclusion}

In this note, we study the controllability of bilinear systems induced by proper Lie group actions. By utilizing Lie group theory, we establish a connection between the underlying Lie algebra of the induced system and the orbits of its group action, and prove that the rank of the Lie algebra is constant on each orbit. Such pattern of the Lie algebra rank enables us to characterize the controllable submanifolds of induced bilinear systems, and also provides a relaxed rank condition that is both sufficient and necessary. This condition also avoids the difficult task of classifying transitive Lie algebras for a given homogeneous space, which is required in most of the existing works.

\appendix

\subsection{Group Actions, Homogeneous Spaces and Orbits}
\label{appd:group_action}

\begin{definition}[Group Actions and Homogeneous Spaces]
  \label{def:lie.group.action}
  A Lie group action of $G$ on $M$ is a smooth map
  \begin{equation}
    \label{eq:lie.group.action}
    \theta:G\times{}M\to{}M,\quad (g,m)\mapsto{}g.m
  \end{equation}
  such that $h.(g.m)=(hg).m$ for any $h,g\in{}G$ and $m\in{}M$. Given an $x\in{}M$, the \emph{orbit} of $x$ is the image of $G\times\{x\}$ under the group action map $\theta$ in \eqref{eq:lie.group.action}:
  \begin{equation}
    \label{eq:orbit}
    G(x):=\theta(G\times{}\{x\}).
  \end{equation}

  Moreover, a $G$-action is called \emph{transitive} if for any pair of $x_1,x_2\in{}M$, there is some $g\in{}G$ such that $g.x_1=x_2$. If $M$ is endowed with a transitive $G$-action, it is called a \emph{homogeneous $G$-space}. 
\end{definition}

\begin{definition}
  \label{def:fundamental.vector.field}
  For a fixed $X\in\mathfrak{g}$ and any given $m\in{}M$, the curve $\exp(tX)$ on $G$ induces a curve $\Phi(t)=\exp(tX).m$ on $M$ by the $G$-action such that $\Phi(0)=m$. Hence the vector field $X\in\mathfrak{g}$ defines a vector field $\theta_{\ast}X$ on $M$ by taking the derivative of $\Phi(t)$ at $t=0$: for each $m\in M$,
  \begin{equation}
    \label{eq:fundamental.vector.field}
    (\theta_{\ast}X)(m)
    :=\frac{\mathrm{d}}{\mathrm{d}t}\Big|_{t=0}\exp(tX).m.
  \end{equation}
  We call $\theta_{\ast}X$ the \emph{fundamental vector field} of $X$.
\end{definition}

For the establishment of our rank condition for bilinear systems in the form of \eqref{eq:induced.system.son}, there is one type of group actions that is of particular interest to us:

\begin{definition}[Proper Actions]
  \label{def:proper.action}
  A $G$-action on $M$ is \emph{proper} if the following group action map
  \begin{equation}
    \label{eq:group.action.map}
    \begin{array}{cccc}
      \phi: & G\times{}M & \to & M\times{}M \\
            & (g,m) & \mapsto & (m,g.{}m)
    \end{array}
  \end{equation}
  is proper, i.e., pre-images of compact subsets in $M\times{}M$ are compact in $G\times{}M$. Particularly, if $G$ is compact, all $G$-actions are proper.
\end{definition}

One important feature of proper actions is that their orbits are \emph{embedded} submanifolds, as stated in the next lemma.

\begin{lemma}[See \cite{ortega2013momentum}]
  \label{lem:orbits.of.a.proper.action} 
  An orbit $\mathcal{O}$ for a proper $G$-action is a closed, embedded submanifold of $M$, whose tangent space has the form
  \begin{equation}
    \mathsf{T}_m\mathcal{O}=\{\theta_{\ast}X(m):X\in\mathfrak{g}\},\quad{}
    \text{for any } m\in{}\mathcal{O}.
  \end{equation}
\end{lemma}

\begin{proof}
  The orbit $\mathcal{O}$ being a closed, embedded submanifold is proved in \cite[Corollary~2.3.33]{ortega2013momentum}, and its tangent space $\mathsf{T}_{m}\mathcal{O}$ can be computed using the constant rank theorem. More specifically, consider an evaluation map $\ev:G\to{}\mathcal{O}$ such that
  \[
    \ev(g):=g.{}m.
  \]
  
  Note that $\ev$ is $G$-equivariant, i.e., $\ev(hg)=h.{}\ev(g)$ for all $h\in{}G$, so it is of constant rank, say $r$, by the equivariant rank theorem \cite{LEE2013INTRO}. Therefore, since $\ev$ is surjective, at the identity there is a basis $\{X_1,\ldots,X_l\}$ of $\mathfrak{g}=\mathsf{T}_{e}G$ such that $\{\mathrm{d}(\ev)(X_i)\}_{i=1}^r$ forms a basis of $\mathsf{T}_m\mathcal{O}$ and that $\mathrm{d}(\ev)(X_j)=0$ for $j=r+1,\ldots,l$. Since
  \begin{align*}
    \mathrm{d}(\ev)(X) &= \frac{\mathrm{d}}{\mathrm{d}t}\Big|_{t=0}\ev(\exp{}tX) \\
      &= \frac{\mathrm{d}}{\mathrm{d}t}\Big|_{t=0}(\exp{}tX).{}m=\theta_{\ast}(X)(m)
  \end{align*}
  for any $X\in\mathfrak{g}$, we conclude that
  \begin{align*}
    \mathsf{T}_{m}\mathcal{O} &=\mathrm{span}\,\{\mathrm{d}(\ev)(X_i)\}_{i=1}^r \\
      &=\mathrm{span}\,\{\mathrm{d}(\ev)(X_i)\}_{i=1}^l \\
      &=\mathrm{span}\,\{\theta_{\ast}(X_i)(m)\} \\
      &=\{\theta_{\ast}(X)(m):X\in\mathfrak{g}\}. \qedhere
  \end{align*}
\end{proof}

\subsection{Lie Algebra Rank Condition}
\label{appd:LARC}

Consider a right-invariant bilinear system defined on a compact, connected Lie group $G$ in the form of
\begin{equation}
  \label{eq:bilinear.system_appendix}
  \dot{X}(t)=B_{0}X(t)+\Bigl(\sum_{i=1}^mu_i(t)B_i\Bigr)X(t),\quad{}X(0)=I,
\end{equation}
where $X(t)\in G$ is the state, $I$ is the identity element of $G$, $B_0$, \dots, $B_m$ are elements in the Lie algebra $\mathfrak{g}$ of $G$, and $u_i(t)\in\mathbb{R}$ are piecewise constant control inputs. Let $\Gamma=\{B_0,\dots,B_m\}$ be the set of the drift and control vector fields of the system in \eqref{eq:bilinear.system_appendix} evaluated at the identity element $I$, then we use $\lie(\Gamma)$ to denote the Lie subalgebra of $\mathfrak{g}$ generated by $\Gamma$, which is the smallest vector subspace of $\mathfrak{g}$ that contains $\Gamma$ and is closed under the Lie bracket operation, i.e., $[C,D]=CD-DC$ for $C,D\in\mathfrak{g}$. The LARC then estabilishes a connection between the controllability of the systems in the form of \eqref{eq:bilinear.system_appendix} and the Lie algebras generated by the vector fields governing the system dynamics.
\begin{theorem}[LARC, see \cite{Brockett72}]
  \label{thm:LARC}
  The system in \eqref{eq:bilinear.system_appendix} is \emph{controllable} on $G$ if and only if $\lie(\Gamma)=\mathfrak{g}$.
\end{theorem}

\bibliographystyle{IEEEtran}
\bibliography{induced,SOn}

\end{document}